\documentclass[11pt]{article}

\usepackage[preprint]{acl}

\usepackage{times}
\usepackage{latexsym}

\usepackage[T1]{fontenc}

\usepackage[utf8]{inputenc}

\usepackage{microtype}

\usepackage{inconsolata}

\usepackage{graphicx}

\usepackage{amsmath,amssymb,amsthm}
\usepackage[ruled,vlined,linesnumbered]{algorithm2e}
\usepackage{booktabs,tabularx}
\usepackage{multirow}
\usepackage{makecell}
\usepackage{siunitx}
\usepackage{adjustbox}
\usepackage{enumitem}
\usepackage{threeparttable}

\newtheorem{lemma}{Lemma}

\usepackage[most]{tcolorbox}
\usepackage{xcolor}

\title{RMCW: A Deletion-Robust Watermark Based on Reed--Muller Codes\\for Language Models}

\author{
Yi Wang$^{1,5}$\thanks{Equal contribution; order decided by a coin flip.},
Baicheng Chen$^{2,4}$\footnotemark[1],
Yu Wang$^{6}$, Jian Zhao$^{1,3}$,
\textbf{Yilei Chen}$^{1,2}$\thanks{Corresponding authors.}, \textbf{Tianxing He}$^{1,3,2}$\footnotemark[2] \\
$^1$Institute for Interdisciplinary Information Sciences, Tsinghua University,\\
$^2$Shanghai Qi Zhi Institute,
$^3$Xiongan AI Institute,\\
$^4$The Chinese University of Hong Kong, Shenzhen,
$^5$University of Wisconsin--Madison,\\
$^6$Institute of Information Engineering, Chinese Academy of Sciences\\
\texttt{wang3883@wisc.edu} \quad \texttt{baichengchen@link.cuhk.edu.cn}\\
\texttt{\textrm{\{}chenyilei, hetianxing\textrm{\}}@mail.tsinghua.edu.cn} \\
}

\newcommand{\method}{RMCW\xspace}
\newcommand{\Sref}[1]{\S\ref{#1}}

\begin{document}
\maketitle
\begin{abstract}
Large Language Model (LLM) watermarking provides a lightweight mechanism for identifying text generated by a specific model, but its robustness remains fragile under post-processing attacks. Deletion attacks are particularly challenging because they shift token positions and break the alignment between observed tokens and their original watermark positions. We propose \textbf{Reed--Muller Code Watermarking (\method)}, an LLM watermarking method based on Reed--Muller codes. In contrast to global codeword recovery, \method searches for surviving local algebraic structure, leveraging the Reed--Solomon consistency induced by affine-line restrictions of Reed--Muller codewords. During generation, \method injects a Reed--Muller structure into the sequence via a secret-keyed vocabulary partition. During detection, it maps the given text to keyed vocabulary bins and tests local subsequences for low-degree Reed--Solomon consistency using Berlekamp--Welch tests. Experiments on C4 and ELI5 datasets with OPT-1.3B and Llama-3.1-8B-Instruct show that \method preserves strong clean-text detectability and outperforms or matches the baseline methods under several deletion and rewriting attacks. Our code is available at \url{https://github.com/BaichengDanny/RMCW}.
\end{abstract}

\section{Introduction}

Large language models (LLMs) support a wide range of tasks, including instruction following~\cite{ouyang2022training}, knowledge-intensive question answering~\citep{lewis2020retrieval,wang2025aicrypto,chen2026crebench}, and code generation~\citep{yang2024swe,chen2026crebench}.
While these systems improve productivity, they also make it difficult to determine whether a passage is written by a human or produced by a model~\citep{chakraborty2023possibilities}. 
LLM watermarking addresses this problem by injecting a hidden signal during generation, allowing a detector to later identify whether a suspect passage is produced by a watermarked model~\citep{kirchenbauer2023watermark,dathathri2024scalable,liu2024survey}.
A practical watermark should remain detectable from a limited amount of text, introduce little degradation in generation quality, and tolerate common post-processing operations~\citep{kirchenbauer2023watermark,kuditipudi2023robust,lalai2025intentions,kirchenbauer2024reliability}.

Error-correcting codes (ECCs) provide a theoretical way for robust watermarking. By introducing structured redundancy, an ECC can preserve a detectable signal even when some embedded symbols are corrupted ~\citep{christ2024pseudorandom,qu2025provably}.
However, when the encoded signal is tied to generation positions, deletions introduce an additional challenge for watermark detection in practice. A deletion removes a token and shifts all subsequent token-derived symbols relative to the original watermark positions~\citep{kirchenbauer2023watermark}. We refer to this deletion-induced loss of position alignment as \emph{synchronization loss}~\citep{lcvenshtcin1966binary}, which is common when generated text is shortened, cropped, or only partially reproduced before
redistribution~\citep{pan2024markllm,christ2024pseudorandom,kuditipudi2023robust}.

\begin{figure*}[t]
    \centering
    \includegraphics[width=\textwidth]{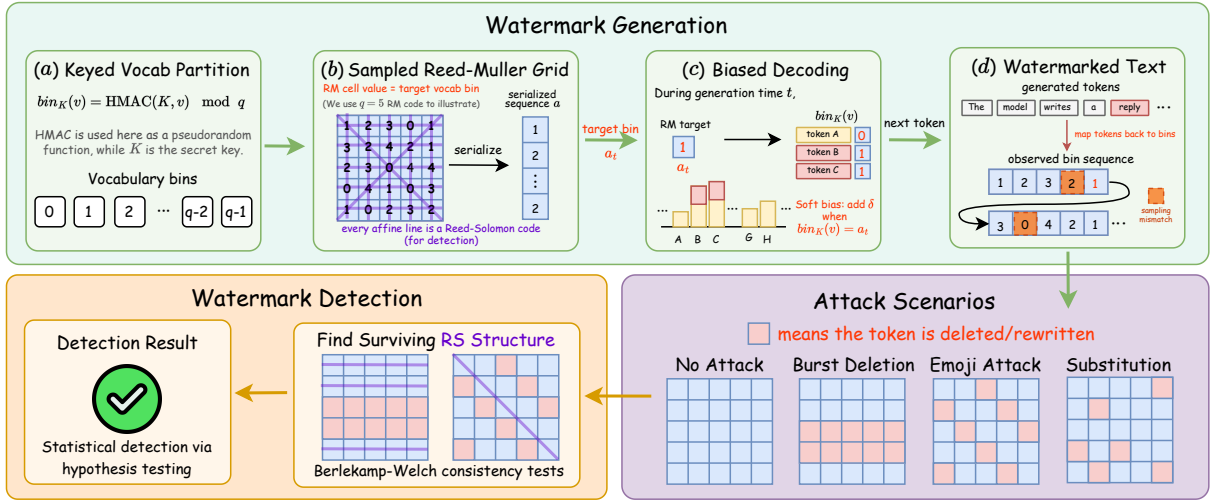}
    \caption{
    Overview of Reed--Muller Code Watermarking (\method).
    \textbf{Generation (top):} First, a secret-keyed vocabulary partition assigns each vocabulary token to one of $q$ bins, while a bivariate Reed--Muller codeword is sampled as a $q\times q$ grid (explained in \Sref{sec:generation}). Then, the sampled grid is serialized into a structured sequence of target bin indices. At each decoding step, the language model increases the logits of tokens assigned to the target bin.
    \textbf{Detection (bottom):} Given candidate text that may have undergone deletion or rewriting, the detector maps each token to its secret-keyed bin label and samples local windows at candidate strides (explained in \Sref{sec:detection}). It tests each window for low-degree Reed--Solomon consistency, aggregates evidence across the tested windows, and outputs the watermark detection decision.
    }
    \label{fig:framework}
    \vspace{-0.5cm}
\end{figure*}

To address the synchronization loss challenge, we propose \textbf{Reed--Muller Code Watermarking (\method)}, an LLM watermarking method based on the Reed--Muller code. Our key observation is that watermark detection is a weaker task than full codeword recovery: the detector only needs to find a structure in the observed sequence that is unlikely to occur in unwatermarked text. Reed--Muller codes are well suited to this goal because their restrictions on affine lines form Reed--Solomon codewords. 
During generation, a secret-keyed vocabulary partition assigns each token a $q$-ary bin label, and the model is biased toward a structured sequence derived from a Reed--Muller codeword.
During detection, \method maps candidate text back to bin labels and searches local subsequences for low-degree Reed--Solomon consistency. 
This local detection design enables the detector to identify surviving algebraic evidence even when deletions break the original position alignment.

We evaluate \method on C4~\citep{raffel2020exploring} and ELI5~\citep{fan2019eli5} using OPT-1.3B~\citep{zhang2022opt} and Llama-3.1-8B-Instruct~\citep{grattafiori2024Llama}. We compare against KGW~\citep{kirchenbauer2023watermark}, EXP~\citep{aronsonpowerpoint}, and PRC~\citep{christ2024pseudorandom} under clean generation and multiple post-processing attacks. 
Our results show that \method preserves strong clean-text detectability, achieving 99.8\% TPR@1\%FPR on unmodified text, while retaining 98.1\% under burst deletion and 86.6\% under synonym substitution. Overall, \method substantially improves upon or remains competitive with the baseline methods across several post-processing attacks.

Our contributions are as follows:
\begin{itemize}[nosep]
    \item We identify deletion-induced synchronization loss as a key challenge for position-aligned coding-based watermarks and formulate detection as local structure testing rather than global codeword recovery.
    \item We propose \method, a new watermarking framework for LLMs that combines Reed--Muller-coded watermark generation with local Reed--Solomon consistency testing.
    \item We theoretically analyze the local consistency test and empirically evaluate RMCW across multiple models, datasets, and post-processing attacks.
\end{itemize}

\section{Related Work}

\paragraph{LLM watermark.} LLM watermarking aims to identify machine-generated text by embedding a hidden signal during generation while preserving text quality~\citep{liu2024survey,liang2026watermarking}. A widely used paradigm is token-level statistical watermarking. KGW~\citep{kirchenbauer2023watermark} partitions the vocabulary into green and red token sets at each generation step, biases decoding toward green tokens, and detects the watermark by testing the fraction of green tokens. 

Follow-up methods improve this paradigm from different perspectives, such as using fixed vocabulary partitions~\citep{zhao2024unigram}, improving detection in low-entropy regions~\citep{lu2024entropy,lee2024wrote}, or making the signal more stable under semantic-preserving edits~\citep{liu2024semantic}. Another line studies distribution-preserving or distortion-free watermarking, where the watermarked sampler is designed to better preserve the original model distribution while still enabling keyed detection~\citep{christ2024undetectable,kuditipudi2023robust}.

Another sequence of work connects watermarking with coding theory~\citep{christ2024pseudorandom,qu2025provably}. PRC watermarking uses pseudorandom error-correcting codes so that watermarked text is difficult to distinguish without the secret key, while a keyed detector can still identify the watermark after text edits~\citep{christ2024pseudorandom}. 
Our work follows this coding-based direction, but focuses on robustness under deletion attacks.

\paragraph{Error-correcting codes.} Error-correcting codes (ECCs) provide structured redundancy for reliable communication under noise. Among classical algebraic codes, Reed--Solomon (RS) codes evaluate low-degree univariate polynomials over finite fields~\citep{ReedSolomon}. Their algebraic structure supports efficient consistency checking and decoding, including the Berlekamp--Welch algorithm for bounded error correction~\citep{welch1986error} and Guruswami--Sudan list decoding~\citep{GuruswamiSudan}. 

Reed--Muller (RM) codes generalize this idea to multivariate low-degree polynomials~\citep{RM-Reed}.  A key property of RM codes is that restricting an RM codeword to any affine line yields an RS codeword. This affine line structure enables local testing of low-degree consistency, which is central to our watermark detector.

\section{Method}

We give an overview of \method in Figure~\ref{fig:framework} and begin with the preliminaries.

\subsection{Preliminaries}
\label{sec:preliminaries}

\noindent We briefly introduce the algebraic properties used in our watermark construction.
Let $\mathbb{F}_q$ be a finite field of prime order $q$.\footnote{In this case, field operations can be understood as arithmetic modulo $q$.} \method uses a polynomial degree bound $d$ and codeword length $m$, which satisfy $0 \leq d < m \leq q$. We also assume the LLM is autoregressive.

\vspace{-2mm}
\paragraph{Reed--Solomon codes.} 
Let $g(z)\in\mathbb F_q[z]$ be a univariate polynomial. Evaluating $g$ on $m$ distinct field elements $0,\ldots,m-1$ gives a length-$m$ codeword vector $\mathbf c(g)$. The degree-$d$ Reed--Solomon (RS) code collects these codewords:
\begin{equation}
\begin{aligned}
    \mathbf c(g)
    &:=
    \bigl(g(0),\ldots,g(m-1)\bigr),\\
    \mathrm{RS}_{q,d,m}
    &:=
    \left\{
        \mathbf c(g):
        g\in\mathbb F_q[z],\ \deg(g)\le d
    \right\}.
\end{aligned}
\label{eq:rs-code}
\end{equation}
Here $m\le q$ makes the inputs distinct, while $d<m$ gives the code redundancy. Therefore, a vector drawn from this code family must agree with a single low-degree polynomial across all coordinates.

\paragraph{Reed--Muller codes.}
\method uses a bivariate polynomial $f(u,v)\in\mathbb F_q[u,v]$ of total degree at most $d$: $f(u,v)=\sum_{i+j\leq d}\beta_{i,j}u^iv^j$, where $\beta_{i,j}\in\mathbb F_q$ are its coefficients. Its evaluations over $\mathbb F_q^2$ form an RM codeword, represented as a $q\times q$ grid $C_f$ with $C_f[u,v]=f(u,v)$ (illustrated in Figure~\ref{fig:framework}(b)).

\paragraph{Affine restrictions.}
RM codes exhibit a key local property on affine lines. Choose a starting point $\mathbf a=(a_1,a_2)\in\mathbb F_q^2$ and a nonzero direction $\mathbf b=(b_1,b_2)\in\mathbb F_q^2$. As $t\in\mathbb F_q$ varies, the points $\mathbf a+t\mathbf b$ form an affine line. Restricting $f$ to this line gives
\begin{equation}
    g(t)
    :=
    f(\mathbf a+t\mathbf b)
    =
    f(a_1+tb_1,a_2+tb_2).
\label{eq:affine-restriction}
\end{equation}
Since both arguments of $f$ are linear in $t$, $g$ is univariate with degree at most $d$. Therefore, if we move along an affine line and record the RM symbol at each visited grid point, the resulting 1D symbol sequence is an RS codeword (illustrated in Figure~\ref{fig:framework}(b)). This property supplies many locally testable low-degree structures inside one RM grid.

\paragraph{Threat model.}
In our setting, an adversary can perform text editing on the output of a watermarked LLM before it reaches the detector. Given the edited text, the detector (\textit{e.g.}, model provider) aims to decide watermark presence without access to the original text and positions, or the correspondence between surviving tokens and RM coordinates.

\begin{table}[t]
\centering
\small
\setlength{\tabcolsep}{4pt}
\renewcommand{\arraystretch}{1.04}
\begin{tabularx}{\columnwidth}{@{}lX@{}}
\toprule
\textbf{Notation} & \textbf{Meaning} \\
\midrule
$q,d,m,e$ & Field size, degree bound, symbols per window, and tolerated mismatches. \\
$K_{\mathrm{bin}},b_{K_{\mathrm{bin}}}$ & Secret binning key and token-to-symbol map. \\
$s_{\mathrm{payload}},f_{\mathrm{pay}},A_{\mathrm{pay}}$ & Payload seed, RM polynomial, and flattened target sequence. \\
$\widehat A,W(s,\tau)$ & Observed symbols and a window with start $s$ and stride $\tau$. \\
$X_\tau,\pi_{0,\tau},p_\tau,p_{\mathrm{det}}$ &
Success count, null rate, and approximate detection scores. \\
\bottomrule
\end{tabularx}
\vspace{-3mm}
\caption{Core notations.}
\label{tab:notation}
\vspace{-5mm}
\end{table}

\subsection{From Global Decoding to Local Detection}
Traditional coding-based watermarking aims to recover an embedded message after adversarial modifications~\citep{christ2024pseudorandom}. This is natural when the received symbols remain aligned with their original positions. However, deletions break this alignment. After deletion, the observed symbol sequence becomes an incomplete and position-shifted realization of the intended watermark pattern.
Thus, the main difficulty is not only that some symbols are missing, but also that the remaining symbols are no longer synchronized with the watermark positions.

For zero-bit watermarking, recovering the original codeword is stronger than necessary. The detector only needs to decide whether a suspect text is drawn from the watermarked distribution. 
Therefore, we shift from global codeword recovery to local structure testing. 
This detection-oriented view is the basis of our deletion-robust design.

\setlength{\textfloatsep}{6pt plus 2pt minus 1pt}

\begin{algorithm}[t]
\caption{Watermark Generation}
\label{alg:generation}
\small
\DontPrintSemicolon
\KwIn{Prompt $x$, language model $M$, parameters $(q,d,\delta,T)$,
binning key $K_{\mathrm{bin}}$, payload seed $s_{\mathrm{payload}}$}
\KwOut{Watermarked text $y$}

Construct $b_{K_\mathrm{bin}}:\mathcal V\rightarrow\mathbb F_q$\;

Use $s_{\mathrm{payload}}$ to construct the degree-$d$ RM polynomial
$f_\mathrm{pay}$ and serialize the resulting codeword as
$A_\mathrm{pay}=(a_0,\ldots,a_{q^2-1})$\;

Initialize the generation context with $x$\;

\For{$t\leftarrow0$ \KwTo $T-1$}{
    Obtain next-token logits $\ell_t(v)$ from $M$\;

    Set $a_t\leftarrow A_\mathrm{pay}[t\bmod q^2]$ and
    $G_t\leftarrow\{v\in\mathcal{V}:b_{K_\mathrm{bin}}(v)=a_t\}$\;

    Bias the logits by
    $\widetilde{\ell}_t(v)\leftarrow
    \ell_t(v)+\delta\cdot\mathbf 1\{v\in G_t\}$\;

    Sample the next token from
    $\mathrm{softmax}(\widetilde{\ell}_t)$ and append it to the context\;
}

\Return generated text $y$\;
\end{algorithm}
\begin{algorithm}[t]
\caption{Watermark Detection}
\label{alg:detection}
\small
\DontPrintSemicolon

\KwIn{Suspect text $y'$, parameters $(q,d,m,e,N_{\mathrm{pilot}},N_{\mathrm{trials}},k,\alpha)$,
binning key $K_{\mathrm{bin}}$, null rates
$\{\pi_{0,\tau}\}_{\tau\in\mathcal T}$}
\KwOut{Detection decision and scores}

Tokenize $y'$ as $(w'_1,\ldots,w'_n)$ and set
$\widehat a_i\leftarrow b_{K_{\mathrm{bin}}}(w'_i)$\;

\If{$n<m$}{
    \Return \textsc{Not Detected} with
    $p_{\mathrm{det}}=1$\;
}

Construct candidate strides $\mathcal{T}=\left\{1,\ldots,\left\lfloor\frac{n-1}{m-1}\right\rfloor\right\}$

Use $N_{\mathrm{pilot}}$ pilot trials to select up to $k$
highest-scoring valid strides
$\mathcal R\subseteq\mathcal T$\;

\ForEach{$\tau\in\mathcal R$}{
    $X_\tau\leftarrow0$\;
    \For{$r\leftarrow1$ \KwTo $N_{\mathrm{trials}}$}{
        Sample a valid start index $s_r^{(\tau)}$ uniformly from $\{1,\ldots,n-(m-1)\tau\}$\;

        Set window $W(s_r^{(\tau)},\tau)$ starting with $s_r^{(\tau)}$\;

        $X_\tau\leftarrow X_\tau+
        \mathrm{Cons}_e(W(s_r^{(\tau)},\tau))$ (using Eq.~\eqref{eq:cons})\;
    }

    Compute $p_\tau$ using Eq.~\eqref{eq:stride-score}\;}

$p_{\mathrm{det}}\leftarrow
\min\{1,|\mathcal R|
\min_{\tau\in\mathcal R}p_\tau\}$\;

\If{$p_{\mathrm{det}}\le\alpha$}{
    \Return \textsc{Detected} with
    $(\{X_\tau,p_\tau\}_{\tau\in\mathcal R},
    p_{\mathrm{det}})$\;
}
\Else{
    \Return \textsc{Not Detected} with
    $(\{X_\tau,p_\tau\}_{\tau\in\mathcal R},
    p_{\mathrm{det}})$\;
}
\end{algorithm}

\subsection{Watermark Generation}
\label{sec:generation}
As illustrated in Algorithm~\ref{alg:generation}, our method follows the standard logits-based watermarking paradigm: the base language model remains fixed, and the watermark is injected by slightly modifying the next token distribution during decoding~\cite{kirchenbauer2023watermark}. The generation procedure has three steps: keyed vocabulary partition, Reed--Muller symbol construction, and logit bias injection. 

\paragraph{Keyed vocabulary partition.} First, let $\mathcal{V}$ be the language model vocabulary and let $v\in\mathcal{V}$ denote a candidate next token. A secret key $K_{\mathrm{bin}}$ defines a keyed partition of $\mathcal{V}$ into $q$ bins:
\begin{equation}
\label{eq:binning}
    b_{K_{\mathrm{bin}}}(v) = \mathrm{HMAC}(K_{\mathrm{bin}}, v) \bmod q \in \mathbb{F}_q.
\end{equation}
Here, $\mathrm{HMAC}$ is used only as a practical keyed pseudorandom function~\citep{krawczyk1997hmac} and $b_{K_{\mathrm{bin}}}(v)\in\mathbb{F}_q$ is the watermark symbol assigned to token $v$. This partition associates each symbol with a keyed subset of the vocabulary, allowing the decoder to preserve lexical flexibility while biasing generation toward a target symbol (illustrated in Figure~\ref{fig:framework}(a)).

\paragraph{Reed--Muller symbol construction.} Next, a payload seed $s_{\mathrm{payload}}$ is expanded pseudorandomly into coefficients $\beta_{i,j}\in\mathbb{F}_q$. These coefficients define a degree-$d$ bivariate polynomial and its $q\times q$ grid:
\begin{equation}
\begin{split}
f_{\mathrm{pay}}(u_1,u_2) &= \sum_{i+j\le d} \beta_{i,j} u_1^i u_2^j, \\
C_{\mathrm{pay}}[u_1,u_2] &= f_{\mathrm{pay}}(u_1,u_2), \quad (u_1,u_2)\in\mathbb{F}_q^2.
\end{split}
\end{equation}
The label $\mathrm{pay}$ indicates dependence on $s_{\mathrm{payload}}$. The grid $C_{\mathrm{pay}}$ represents an RM codeword. Row-major serialization gives $A_{\mathrm{pay}}=(a_0,\cdots,a_{q^2-1})$. At generation step $t$, the target symbol is $a_t=A_{\mathrm{pay}}[t \mod q^2]$, which repeats or truncates the base sequence according to the requested generation length $T$ (illustrated in Figure~\ref{fig:framework}(b)).

\paragraph{Logit bias injection.} Finally, let $\ell_t(v)$ be the original logit for candidate token $v$ at step $t$. We increase the logits of tokens assigned to the target symbol:
\begin{equation}
    \widetilde{\ell}_t(v)=\ell_t(v)+\delta\cdot\mathbf{1}\{b_{K_{\mathrm{bin}}}(v)=a_t\}, \delta\geq0.
\end{equation}
The next token $y_{t+1}$ is sampled from the distribution induced by $\widetilde{\ell}_t$ (illustrated in Figure~\ref{fig:framework}(c)).

\subsection{Watermark Detection}
\label{sec:detection}
As shown in Algorithm~\ref{alg:detection}, given a suspect text, the detector tests whether its induced symbol sequence contains local algebraic structure generated by the secret RM codeword and aggregates the evidence through a calibrated hypothesis test.

\paragraph{Token-to-symbol conversion.} First, the detector tokenizes the suspect text $y'$ as $(w'_1,\ldots,w'_n)$, where $w_i'\in\mathcal{V}$. Applying the same keyed partition as in Eq.~\eqref{eq:binning} gives:
\begin{equation}
    \hat{A}=(\hat{a}_1,\cdots,\hat{a}_n),\quad \hat{a}_i=b_{K_{\mathrm{bin}}}(w'_i)\in \mathbb{F}_q.
\end{equation}
The remainder of the detector operates only on $\hat{A}$.\footnote{Detection requires access to the secret binning key $K_{\mathrm{bin}}$ to reconstruct the same keyed vocabulary partition.}

\paragraph{Local consistency test.} Next, for a start index $s$, a window length $m$, and a stride $\tau$, we first define the window to be
\begin{equation}
\begin{split}
    W(s,\tau) = \bigl( & \hat{a}_{s}, \hat{a}_{s+\tau}, \ldots, \hat{a}_{s+(m-1)\tau} \bigr), \\
    & s \in \bigl\{ 1, \ldots, n-(m-1)\tau \bigr\},
\end{split}
\end{equation}
which guarantees that all indices are valid.

For each window, the detector checks whether it is consistent with a degree $d$ univariate polynomial over $\mathbb{F}_q$ (refer to \Sref{sec:preliminaries}).
We use the following error-tolerant consistency criterion:
\begin{equation}
\label{eq:cons}
\begin{gathered}
\mathrm{Cons}_e(W(s,\tau))=1 \Longleftrightarrow \exists g\in\mathbb{F}_q[z], \\
\deg(g)\le d,\ 
\left|\{j: W_j(s,\tau)\neq g(j)\}\right|\le e.
\end{gathered}
\end{equation}
Here, $e$ is the tolerated mismatch budget. We evaluate the predicate using Berlekamp--Welch decoding followed by explicit distance verification~\cite{welch1986error}. Setting $e=0$ gives exact RS consistency, while $e>0$ tolerates a limited number of substitutions or bin mismatches.

\paragraph{Adaptive stride selection.} The detector first constructs a candidate set of all valid strides $\mathcal{T}=\left\{1,\ldots,\left\lfloor\frac{n-1}{m-1}\right\rfloor\right\}$.
A candidate set $\mathcal{T}$ contains the strides searched by the detector. For each valid $\tau\in\mathcal{T}$, a pilot stage samples $N_{\mathrm{pilot}}$ windows and counts their consistency successes. The detector retains up to $k$ highest-scoring valid strides in $\mathcal{R}\subseteq\mathcal{T}$ for the full test. This step is a heuristic search for promising windows, and lets the detector adapt to different deletion patterns without assuming a fixed edit structure in advance.

\paragraph{Statistical decision score.} After pilot selection, the detector evaluates each shortlisted stride $\tau\in\mathcal{R}$ using $N_{\mathrm{trials}}$ sampled windows.
Let $s_r^{(\tau)}$ be the start index sampled in the $r$-th trial and $X_\tau$ be the number of sampled windows that pass the local consistency test for stride $\tau$.
Define
\begin{equation}
\label{eq:stride-score}
\begin{split}
    X_\tau &= \sum_{r=1}^{N_\mathrm{trials}}\mathrm{Cons}_e(W(s_r^{(\tau)},\tau)), \\
    p_\tau &=\Pr_{Y\sim\mathrm{Binom}(N_\mathrm{trials},\pi_{0,\tau})} \left[Y\ge X_\tau\right].
\end{split}
\end{equation}
Here, $\pi_{0,\tau}$ is the reference null probability that one sampled window passes the consistency test under stride $\tau$. We aggregate the strongest shortlisted result using Bonferroni correction:
\begin{equation}
p_{\mathrm{det}}=\min\left\{1,\ |\mathcal{R}|\cdot \min_{\tau\in\mathcal{R}}p_\tau\right\}.
\end{equation}
The detector declares the text as watermarked if $p_{\mathrm{det}}\le\alpha$, where $\alpha$ is the target false positive level.

\begin{table*}[t]
\centering
\small
\setlength{\tabcolsep}{4.5pt}
\begin{tabular}{llcccccccc}
\toprule
\multirow{2}{*}{Method} & \multirow{2}{*}{Metric}
& \multicolumn{8}{c}{Attack Scenarios} \\
\cmidrule(lr){3-10}
& & Clean & Rand. Del. & Burst Del. & Trunc. & Syn. Sub. & Para. & Emoji & Avg. \\
\midrule

\multirow{3}{*}{KGW}
& TPR@1\%FPR & 99.7 & \textbf{98.2} & 96.4 & 97.2 & 69.5 & \textbf{12.2} & 0.0 & 67.6 \\
& TPR@5\%FPR & \textbf{99.8} & \textbf{99.3} & \textbf{98.9} & \textbf{99.1} & 85.3 & 15.3 & 0.0 & 71.1 \\
& AUROC       & \textbf{99.9} & \textbf{99.8} & \textbf{99.8} & \textbf{99.8} & 97.3 & 41.2 & 28.6 & 80.9 \\
\midrule

\multirow{3}{*}{EXP}
& TPR@1\%FPR & 98.9 & 79.1 & 93.8 & 94.6 & 21.1 & 10.1 & 1.9 & 57.1 \\
& TPR@5\%FPR & 99.5 & 87.3 & 97.3 & 96.8 & 33.3 & \textbf{15.6} & 5.6 & 62.2 \\
& AUROC       & 99.6 & 96.9 & 99.2 & 99.0 & 75.9 & 47.5 & 48.4 & 80.9 \\
\midrule

\multirow{3}{*}{PRC}
& TPR@1\%FPR & 95.0 & 55.2 & 45.9 & 39.7 & 29.9 & 8.4 & 0.0 & 39.2 \\
& TPR@5\%FPR & 97.6 & 55.9 & 47.1 & 40.4 & 36.2 & 9.6 & 0.2 & 41.0 \\
& AUROC       & 98.9 & 56.2 & 47.5 & 40.7 & 42.4 & 32.7 & 33.1 & 50.2 \\
\midrule

\multirow{3}{*}{\method (Ours)}
& TPR@1\%FPR & \textbf{99.8} & 90.0 & \textbf{98.1} & \textbf{98.9} & \textbf{86.6} & 8.3 & \textbf{95.5} & \textbf{82.5} \\
& TPR@5\%FPR & \textbf{99.8} & 90.4 & 98.8 & 99.0 & \textbf{94.2} & 12.1 & \textbf{95.7} & \textbf{84.3} \\
& AUROC       & \textbf{99.9} & 94.9 & 99.0 & 99.4 & \textbf{98.5} & \textbf{53.7} & \textbf{97.7} & \textbf{91.9} \\
\bottomrule
\end{tabular}
\vspace{-2mm}
\caption{
Main results on C4 with Llama-3.1-8B-Instruct for 500-token generations. We report TPR at \(1\%\) and \(5\%\) false-positive rates (FPRs) and AUROC, all expressed as percentages, on unmodified text (\textsc{Clean}) and under six post-processing attacks. \textsc{Rand. Del.}, \textsc{Burst Del.}, \textsc{Trunc.}, \textsc{Syn. Sub.}, \textsc{Para.}, and \textsc{Emoji} denote random deletion, burst deletion, truncation, synonym substitution, paraphrasing, and emoji attack, respectively. \textsc{Avg.} is the arithmetic mean across the clean setting and all six attack scenarios. Higher values indicate stronger watermark detectability. The best result for each metric under each scenario, including the average, is shown in bold.
}
\label{tab:main_results}
\vspace{-0.5cm}
\end{table*}

\subsection{Theoretical Analysis}
The theoretical foundation of our construction relies on the local RS structure induced by affine restrictions of RM codes. In particular, restricting an RM codeword to any affine line yields a low-degree univariate polynomial, allowing the use of classical algebraic decoding techniques.

\paragraph{Deletion robustness.}
Unlike classical codeword-recovery settings, \method performs detection rather than reconstruction. The detector only needs to find a surviving local subsequence that is consistent, up to a bounded number of errors, with a low-degree RS codeword. Since every affine-line restriction of an RM codeword yields such an RS codeword, deleting part of the generated text may still leave locally detectable algebraic evidence.

RS codes naturally tolerate partial observations, and the remaining symbols of a line can continue to exhibit low-degree consistency. Therefore, for \method, deletions mainly reduce the number of observable coordinates rather than eliminating the watermark structure itself. Even when burst deletions remove an entire region, other windows may retain detectable RS structure.

\paragraph{Berlekamp--Welch decoding.}
To formalize this test, we use the Berlekamp--Welch algorithm, which reconstructs a low-degree polynomial from corrupted outputs.

\noindent\textbf{Proposition 1 (Berlekamp--Welch Algorithm).}
\emph{Let $S=\{\alpha_1,\ldots,\alpha_m\}\subseteq\mathbb F_q$ contain distinct field elements, and let $g\in\mathbb F_q[z]$ satisfy $\deg(g)\le d$. Suppose the received values
$y_1,\ldots,y_m$ differ from $(g(\alpha_1),\ldots,g(\alpha_m))$ in at most $e$ coordinates. If $2e<m-d$, then Berlekamp--Welch uniquely reconstructs $g$ in polynomial time.}

This guarantee provides the algebraic consistency test used by \method to identify local RS structure in the presence of substitution noise.

\paragraph{Substitution robustness.}
Let
\begin{equation}
    \begin{split}
    \mathrm{RS}_{q,d}(S)
    = \bigl\{ &(g(\alpha_1),\ldots,g(\alpha_m)) : \\
    & g\in\mathbb F_q[z],\ \deg(g)\le d \bigr\}
    \end{split}
    \label{eq:rs-evaluation-set}
\end{equation}
be the degree-$d$ RS code over the set $S$. Suppose a codeword is transmitted through a substitution channel in which each coordinate is independently corrupted with probability $\rho$. Let $\mathcal T_{\mathrm{wm}}$ denote the resulting noisy distribution and let $\mathcal T_{\mathrm{null}}$ be the uniform distribution over $\mathbb F_q^m$. The detector accepts when Berlekamp--Welch succeeds with decoding radius $e$.

\noindent\textbf{Theorem 1 (Substitution Bound).}
\emph{Assume $2e<m-d$. Define}
\begin{equation}
    P_{\mathrm{acc}}^{\mathrm{wm}}
    =
    \sum_{i=0}^{e}
    \binom{m}{i}\rho^i(1-\rho)^{m-i}.
\label{eq:wm-acceptance}
\end{equation}
\emph{This is the probability that a noisy RS codeword contains at most $e$ substitutions. Define}
\begin{equation}
    P_{\mathrm{acc}}^{\mathrm{null}}
    =
    \frac{
        \sum_{i=0}^{e}\binom{m}{i}(q-1)^i
    }{q^{m-d-1}}.
\label{eq:null-acceptance}
\end{equation}
\emph{This is the probability that a uniformly random vector lies within Hamming distance at most $e$ of an RS codeword. If}
\begin{equation}
    0.99P_{\mathrm{acc}}^{\mathrm{wm}}
    >
    P_{\mathrm{acc}}^{\mathrm{null}},
\label{eq:substitution-separation}
\end{equation}
\emph{then $\mathcal T_{\mathrm{wm}}$ and
$\mathcal T_{\mathrm{null}}$ are distinguishable with constant advantage.}

The theorem gives a direct parameter condition for substitution-robust detection. Its proof is provided in Appendix~\ref{app:theorem1}.

\section{Experiments}
\label{sec:experiment}

\subsection{Experimental Setup} 

\paragraph{Models.} We conduct experiments on two open-weight language models: OPT-1.3B~\citep{zhang2022opt} and Llama-3.1-8B-Instruct~\citep{grattafiori2024Llama}. These two models differ in model scale and tokenizer design, which allows us to evaluate whether the watermarking behavior remains consistent across model families.

\paragraph{Baseline methods.} We compare \method against KGW~\citep{kirchenbauer2023watermark}, EXP~\citep{aronsonpowerpoint}, and PRC~\citep{christ2024pseudorandom}. KGW is the standard token-level green-list watermark, EXP is a sampling-based watermark, and PRC is the closest ECC-based watermark baseline. Implementation details are provided in Appendix~\ref{app:imp-baseline}.

\paragraph{Datasets.} We evaluate all watermarking methods on two generation settings with a maximum generation length of 500: open-ended text generation and long-form question answering. For open-ended generation, we use the RealNews subset of C4~\citep{raffel2020exploring}, which is widely used in LLM watermarking studies. For long-form question answering, we use ELI5~\citep{fan2019eli5}.

\paragraph{Attacks.} We evaluate watermark robustness under both deletion-based and rewriting-based attacks~\citep{pan2024markllm}. For deletion-based attacks, we consider random deletion, burst deletion, partial truncation, and emoji attack. For rewriting-based attacks, we consider synonym substitution and paraphrasing (with DIPPER paraphraser~\citep{krishna2023paraphrasing}). Details are provided in Appendix~\ref{app:attacks}.

\paragraph{Evaluation metrics.} We report TPR@1\%FPR, TPR@5\%FPR, and AUROC~\citep{gu2026ssg}. Watermarked generations are treated as positive samples, and human-written continuations provided in the datasets are treated as negative samples. TPR@1\%FPR and TPR@5\%FPR measure detection sensitivity under strict false-positive constraints, while AUROC measures the overall separation between watermarked and unwatermarked texts under detection thresholds.

For more detailed experimental settings, please refer to Appendix~\ref{app:setup}.

\subsection{Main Results}
\label{sec:main_results}

Table~\ref{tab:main_results} reports the main robustness results on C4 with Llama-3.1-8B-Instruct (see \Sref{sec:cross} for results on the other model and dataset).
\textbf{Overall, \method maintains strong detectability on unmodified text while providing substantial robustness gains under post-processing attacks.}
On clean generations, \method obtains 99.8\% TPR@1\%FPR and 99.9\% AUROC, comparable to KGW and outperforming EXP and PRC.
When averaged across the clean setting and all six attack scenarios, \method achieves the best performance under all three evaluation metrics, with 82.5\% TPR@1\%FPR, 84.3\% TPR@5\%FPR, and 91.9\% AUROC.
These results exceed the strongest baseline averages by 14.9, 13.2, and 11.0 percentage points, respectively, demonstrating \method leads to robustness gains under post-processing attacks.

\paragraph{Deletion-based attacks.} For random deletion, \method achieves the second-highest TPR@1\%FPR of 90.0\%, outperforming EXP and PRC, although it remains below KGW. It also obtains an AUROC of 94.9\%. The advantage of \method is more pronounced under structured deletions: it achieves 98.1\% TPR@1\%FPR under burst deletion and 98.9\% under truncation, whereas PRC obtains only 45.9\% and 39.7\%, respectively. \method is also highly robust to the emoji attack, retaining 95.5\% TPR@1\%FPR, compared with at most 1.9\% for the baselines. 
These results support our central design intuition that local algebraic evidence can remain detectable when global token-to-codeword alignment is disrupted.

\paragraph{Rewriting-based attacks.} \method achieves 86.6\% TPR@1\%FPR under synonym substitution, substantially outperforming the strongest baseline result of 69.5\%. However, strong paraphrasing remains challenging for all evaluated methods. Although \method achieves the highest AUROC of 53.7\%, its TPR@1\%FPR is only 8.3\%, indicating that extensive semantic rewriting can remove most of the detectable watermark signal. Overall, \method is most effective under structured deletion, token substitution, and severe tokenization disruption, while paraphrasing remains its main limitation.

\paragraph{Implementation details.} We use $q=17$, $d=3$, $m=12$, and $e=2$. The detector uses 500 pilot trials and retains up to 8 strides. We set $\delta=5$ based on perplexity relative to human and baselines.

\begin{table*}[t]
    \centering
    \small
    \setlength{\tabcolsep}{7pt}
    \renewcommand{\arraystretch}{1.12}

    \begin{adjustbox}{max width=0.7\textwidth}
    \begin{tabular}{
        @{}
        l
        l
        l
        S[table-format=2.1]
        S[table-format=2.1]
        S[table-format=2.1]
        @{}
    }
        \toprule
        \textbf{Model}
        & \textbf{Dataset}
        & \textbf{Attack}
        & \multicolumn{1}{c}{\textbf{TPR@1\%FPR}}
        & \multicolumn{1}{c}{\textbf{TPR@5\%FPR}}
        & \multicolumn{1}{c}{\textbf{AUROC}} \\
        \midrule

        \multirow{6}{*}{Llama-3.1-8B-Instruct}
        & \multirow{3}{*}{C4}
        & Clean      & 99.8 & 99.8 & 99.9 \\
        &
        & Burst Del. & 98.1 & 98.8 & 99.0 \\
        &
        & Syn. Sub.  & 86.6 & 94.2 & 98.5 \\

        \cmidrule(lr){2-6}

        &
        \multirow{3}{*}{ELI5}
        & Clean      & 99.5 & 99.5 & 99.8 \\
        &
        & Burst Del. & 95.5 & 95.7 & 97.7 \\
        &
        & Syn. Sub.  & 90.8 & 92.7 & 94.9 \\

        \midrule

        \multirow{6}{*}{OPT-1.3B}
        & \multirow{3}{*}{C4}
        & Clean      & 87.1 & 87.6 & 93.5 \\
        &
        & Burst Del. & 83.6 & 86.8 & 90.6 \\
        &
        & Syn. Sub.  & 81.4 & 83.7 & 87.3 \\

        \cmidrule(lr){2-6}

        &
        \multirow{3}{*}{ELI5}
        & Clean      & 89.6 & 90.4 & 94.7 \\
        &
        & Burst Del. & 87.6 & 87.8 & 92.9 \\
        &
        & Syn. Sub.  & 85.7 & 88.3 & 89.6 \\

        \bottomrule
    \end{tabular}
    \end{adjustbox}

    \caption{Cross-model and cross-dataset detection performance of \method on clean text (\textsc{Clean}) and under burst deletion (\textsc{Burst Del.}) and synonym substitution (\textsc{Syn. Sub.}). We report TPR@1\%FPR, TPR@5\%FPR, and AUROC, all expressed as percentages, on C4 and ELI5 with Llama-3.1-8B-Instruct and OPT-1.3B.}
    \label{tab:generalization}
    \vspace{-4mm}
\end{table*}

\subsection{Analysis}



\subsubsection{Cross Model and Dataset Generalization}
\label{sec:cross}

Table~\ref{tab:generalization} evaluates \method across two model families and two generation settings on clean text, burst deletion, and synonym substitution. With Llama-3.1-8B-Instruct on C4, \method achieves 99.8\% TPR@1\%FPR on clean text and retains 98.1\% and 86.6\% under burst deletion and synonym substitution. The same pattern holds on ELI5, where the corresponding TPRs are 99.5\%, 95.5\%, and 90.8\%. These results show that the watermark remains strongly detectable when both the generation domain and the post-processing operation change.

\method also transfers consistently to OPT-1.3B. On C4, its TPR@1\%FPR is 87.1\% on clean text, 83.6\% under burst deletion, and 81.4\% under synonym substitution. On ELI5, the corresponding results are 89.6\%, 87.6\%, and 85.7\%. Although the absolute detection strength varies across base models, the robustness pattern is preserved in all four model--dataset combinations. 

Overall, these results indicate that \method exhibits consistent robustness across the evaluated model families and generation domains.

\begin{figure}
    \centering
    \includegraphics[width=0.9\linewidth]{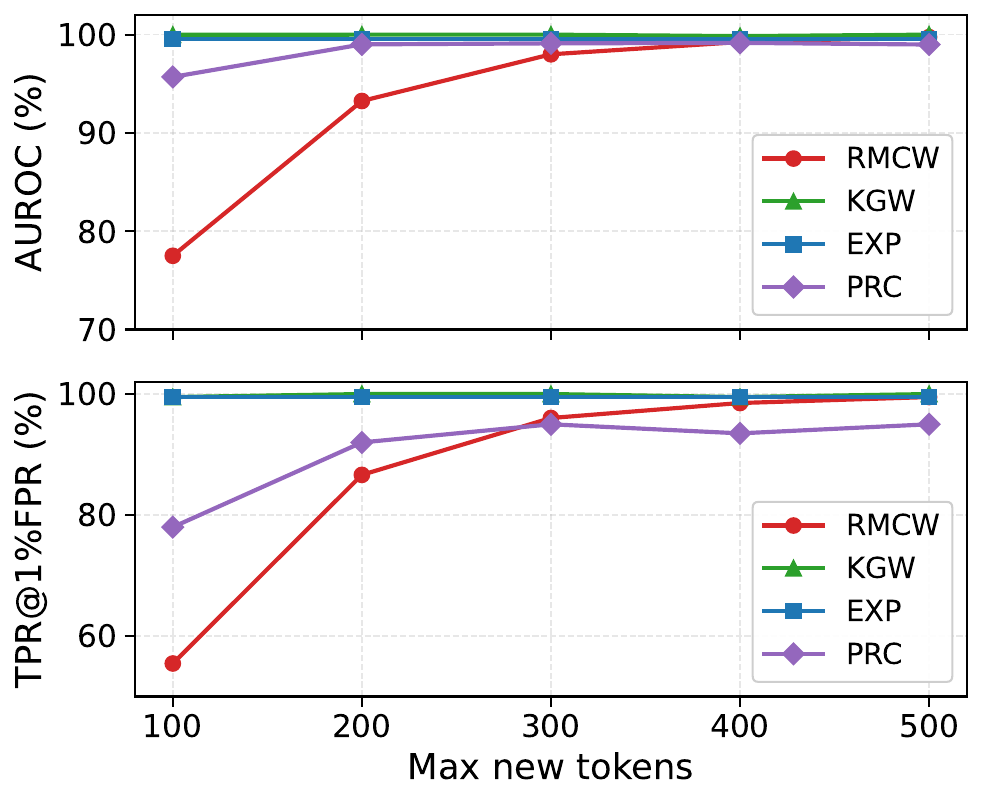}
    \vspace{-4mm}
    \caption{Effect of generation length on clean-text watermark detection on C4 with Llama-3.1-8B-Instruct. AUROC (top) and TPR@1\%FPR (bottom) are reported for \method, KGW, EXP, and PRC as the maximum number of generated tokens increases from 100 to 500.}
    \label{fig:token-score}
\end{figure}

\subsubsection{Effect of Generation Length}

Figure~\ref{fig:token-score} studies how detection performance changes with the maximum generation length. \method benefits consistently from longer outputs. Its AUROC increases from 77.5\% at 100 tokens to above 93.3\% at 200 tokens and approaches 98.0\% from 300 tokens onward. TPR@1\%FPR follows the same trend, increasing from 55.5\% at 100 tokens to 87.6\% at 200 tokens and 96.0\% at 300 tokens before approaching 99.8\% at 500 tokens.

This scaling behavior follows directly from the local detection design. Longer generations provide more induced symbols and more candidate windows, increasing the probability that the detector observes sufficient surviving RS structure. Aggregating evidence over these windows then produces stronger separation between watermarked and unwatermarked text. This observation confirms that local algebraic evidence becomes increasingly reliable as more text is observed. More results and analysis are provided in Appendix~\ref{app:further_analysis}.

\section{Conclusion}

We introduce \method, a deletion-robust Reed--Muller code watermarking method that addresses synchronization loss by testing local algebraic structure. By embedding Reed--Muller-structured symbols and detecting surviving Reed--Solomon consistency, \method remains detectable when post-processing attacks disrupt global token positions. Experiments across models and datasets demonstrate strong clean-text detectability and robustness under deletion and rewriting attacks, establishing local algebraic consistency as a promising design for robust LLM watermarking.

\section*{Limitations}
\paragraph{Robustness boundary.}
\method is designed to detect surviving local algebraic structure after deletion-oriented post-processing, rather than to provide universal robustness against arbitrary text transformations. Its effectiveness therefore depends on some local token-derived structure remaining in the observed text. Strong document-level paraphrasing can replace most of the original lexical realization and sentence structure, leaving little directly surviving watermark evidence. Independently distributed random deletion is also more challenging than burst deletion or truncation because it introduces synchronization changes throughout the sequence.

\paragraph{Dependence on text length.}
The detector benefits substantially from longer text. Local Reed--Solomon consistency testing requires enough observed symbols to form candidate windows and accumulate statistical evidence. As a result, short generations or heavily shortened fragments may contain insufficient evidence for reliable detection at low false-positive rates. Improving short-text detection, potentially through more sample-efficient local tests or evidence aggregation across multiple passages, remains an important direction.

\paragraph{Calibration and deployment.}
Our implementation relies on several calibrated and configuration-dependent quantities. The stride-specific null success rates are estimated offline, and the resulting $p_{\mathrm{det}}$ is an approximate detection score rather than an exactly calibrated $p$-value. Changes in the base model, tokenizer, vocabulary, text domain, generation length, or detector parameters may alter the null distribution and require recalibration. Although the empirical low-FPR metrics are computed consistently in our experiments, practical deployment would require validation on the target data distribution.

\section*{Acknowledgment}
The research is supported by Shanghai Qi Zhi Institute Innovation Program.

\bibliography{custom}

\appendix

\section{LLM Usage Statement}
\label{app:llm_usage}
We use LLMs only as writing and coding assistants during the preparation of this work. In particular, they are used to help refine parts of the codebase and improve the clarity and presentation of the manuscript. All core ideas, theoretical formulations, methodological design, experimental setup, initial code implementation, and analysis are developed and verified by the authors.

\section{Proof}
\label{app:proof}

\subsection{Soundness of Berlekamp--Welch verification}

\begin{lemma}[Soundness of Berlekamp--Welch verification]
Let $\mathbb{F}_q$ be a finite field, and let
\[
S=\{\alpha_1,\ldots,\alpha_m\}\subseteq \mathbb{F}_q
\]
be a set of $m$ distinct evaluation points. Let
\[
\mathrm{RS}_{q,d}(S)
=
\left\{
\big(f(\alpha_1),\ldots,f(\alpha_m)\big)
\;\middle|\; \deg f\le d
\right\}
\]
be the Reed--Solomon code of degree at most $d$ on $S$.

Suppose the Berlekamp--Welch decoder is run with error radius $e$, and its output is accepted only if the recovered polynomial agrees with the received word in at least $m-e$ positions. If
\[
2e < m-d,
\]
then for a uniformly random word
\[
Y\sim \mathbb{F}_q^m,
\]
we have

\[
\Pr\left[
Y \text{ is accepted}
\right]
=
q^{-(m-d-1)}
\sum_{i=0}^{e}\binom{m}{i}(q-1)^i .
\]
In particular, for large $q$ this is approximately
\[
\binom{m}{e}q^{d+1+e-m}.
\]
\end{lemma}

\begin{proof}
The Berlekamp--Welch decoder with verification accepts a received word
\[
Y=(Y_1,\ldots,Y_m)\in \mathbb{F}_q^m
\]
if and only if there exists a polynomial $f\in \mathbb{F}_q[x]$ with
\[
\deg f\le d
\]
such that
\[
\left|
\left\{
i\in[m]\;:\; Y_i\neq f(\alpha_i)
\right\}
\right|
\le e.
\]
Equivalently, $Y$ must lie within Hamming distance at most $e$ from some Reed--Solomon codeword.

First, the number of Reed--Solomon codewords is
\[
|\mathrm{RS}_{q,d}(S)|=q^{d+1},
\]
because a polynomial of degree at most $d$ is specified by its $d+1$ coefficients, and since $m>d$, distinct such polynomials give distinct evaluations on $S$.

For a fixed codeword $c\in \mathrm{RS}_{q,d}(S)$, the number of words within Hamming distance at most $e$ from $c$ is
\[
\sum_{i=0}^{e}\binom{m}{i}(q-1)^i.
\]
Indeed, to form a word at distance exactly $i$ from $c$, we choose the $i$ corrupted positions in $\binom{m}{i}$ ways, and at each chosen position choose one of the $q-1$ symbols different from the original symbol.

It remains to justify that these Hamming balls do not overlap. The Reed--Solomon code has minimum distance
\[
\Delta = m-d.
\]
This follows because if $f$ and $g$ are two distinct polynomials of degree at most $d$, then $f-g$ is a nonzero polynomial of degree at most $d$, and hence has at most $d$ roots. Therefore, $f$ and $g$ can agree on at most $d$ evaluation points, so their corresponding codewords differ in at least $m-d$ positions.

Since
\[
2e < m-d,
\]
the Hamming balls of radius $e$ around distinct codewords are disjoint. Therefore, the total number of accepted words is exactly
\[
q^{d+1}
\sum_{i=0}^{e}\binom{m}{i}(q-1)^i.
\]
Since $Y$ is uniformly random over $\mathbb{F}_q^m$, the acceptance probability is this quantity divided by $q^m$:
\[
\Pr[Y\text{ is accepted}]
=
\frac{
q^{d+1}
\sum_{i=0}^{e}\binom{m}{i}(q-1)^i
}{
q^m
}.
\]
This proves the exact formula.

Finally, when $q$ is large and the largest term in the Hamming ball volume dominates, we have
\[
\sum_{i=0}^{e}\binom{m}{i}(q-1)^i
\approx
\binom{m}{e}q^e.
\]
Thus
\[
\begin{split}
\Pr[Y\text{ is accepted}]
&\approx
q^{d+1-m}\binom{m}{e}q^e \\
&=
\binom{m}{e}q^{d+1+e-m}.
\end{split}
\]
\end{proof}

\subsection{Proof of Theorem 1}
\label{app:theorem1}
\begin{proof}
Under the noisy watermark distribution $\mathcal{T}_{\mathrm{wm}}$, BW decoding succeeds whenever the number of substitutions does not exceed $e$. Since substitutions occur independently with probability $\rho$, the number of corruptions follows a binomial distribution:
\[
X\sim\mathrm{Binomial}(m,\rho).
\]

Thus,
\[
P_{\mathrm{acc}}^{\mathrm{wm}}
=
\Pr[X\le e]
=
\sum_{i=0}^{e}
\binom{m}{i}
\rho^i(1-\rho)^{m-i}.
\]

According to Lemma 1, the acceptance probability of the null distribution $\mathcal{T}_{\mathrm{null}}$ is
\[
\begin{split}
P_{\mathrm{acc}}^{\mathrm{null}}
&=
\frac{
q^{d+1}
\sum_{i=0}^{e}
\binom{m}{i}(q-1)^i
}{
q^m
}\\
&=
\frac{
\sum_{i=0}^{e}
\binom{m}{i}(q-1)^i
}{
q^{m-d-1}
}.
\end{split}
\]

Therefore, whenever the acceptance probability under $\mathcal{T}_{\mathrm{wm}}$ exceeds that under $\mathcal{T}_{\mathrm{null}}$ by a constant factor, the two distributions become distinguishable.
\end{proof}

\section{Detailed Experimental Setup}
\label{app:setup}

\subsection{Datasets}
\label{app:datasets}
We evaluate \method on C4 RealNews and ELI5, which represent two complementary long-form generation settings. For every example, the prompt is provided to the base language model, which generates a continuation with watermarking. The preprocessed MarkLLM files pair each prompt with a human-written continuation. We use these human continuations as the negative set for detection, keeping the negative examples fixed across watermarking methods and attack settings.

\paragraph{C4 RealNews.}
C4 RealNews~\citep{raffel2020exploring} is used for open-ended, news-style text continuation. Its RealNews-like configuration contains 13,804,817 training examples and 13,855 validation examples. In our experiment, we use 5,000 samples from the validation set to evaluate all the watermarking methods. Each example provides a natural-text prefix as the prompt and a human-written continuation. C4 serves as the primary dataset for the main robustness evaluation, the deletion-rate analysis, the quality--detectability analysis, and the generation-length analysis. Its human continuations are also used as the human-text reference in the perplexity comparison.

\paragraph{ELI5.}
ELI5~\citep{fan2019eli5} is used for long-form question answering and contains approximately 270,000 question threads. In our experiment, we use 1,000 samples to evaluate \method. Each question is provided as the prompt, and the model generates an explanatory answer. This setting complements document continuation by testing the watermark on open-ended responses conditioned on information-seeking questions.

\subsection{Models}
\label{app:models}
We conduct experiments with OPT-1.3B~\citep{zhang2022opt} and Llama-3.1-8B-Instruct~\citep{grattafiori2024Llama}. The models differ in scale,
architecture family, and tokenizer design, allowing us to evaluate whether \method transfers across different language-model distributions. Generation is implemented using Hugging Face Transformers. All methods use the same base model, tokenizer, prompt set, generation length, and attack configurations within each experimental condition.

\paragraph{OPT-1.3B.}
OPT-1.3B is a 1.3-billion-parameter decoder-only language model from the Open Pre-trained Transformer family. OPT models have been widely used in prior LLM watermarking evaluations, making OPT-1.3B a useful reference model for comparison with existing work. Its relatively compact scale also provides a distinct generation distribution for evaluating the portability of the watermarking method.

\paragraph{Llama-3.1-8B-Instruct.}
Llama-3.1-8B-Instruct is an instruction-tuned 8-billion-parameter model from the Llama 3.1 family. Compared with OPT-1.3B, it provides a larger and more recent model with a different tokenizer and stronger long-form generation capability. We include it to evaluate \method under a modern instruction-following model and to test whether the same watermarking construction remains effective across changes in model scale and vocabulary.

\subsection{Metrics}
\label{app:metrics}

\paragraph{Detection metrics.} 
The positive set contains watermarked generations, optionally after an attack, and the negative set contains the corresponding human-written continuations from the same dataset split. Given a detection threshold, the true-positive rate and false-positive rate are
\begin{equation}
    \mathrm{TPR}=\frac{\mathrm{TP}}
    {\mathrm{TP}+\mathrm{FN}},
    \mathrm{FPR}=\frac{\mathrm{FP}}
    {\mathrm{FP}+\mathrm{TN}}.
\end{equation}
Following \citet{gu2026ssg}, we report TPR@1\%FPR and TPR@5\%FPR in \Sref{sec:experiment}. And we further report AUROC. AUROC measures the overall ranking
of watermarked and negative samples across all thresholds. For a target
false-positive rate $\gamma$, we compute
\begin{equation}
    \mathrm{TPR}@\gamma\mathrm{FPR}
    = \max_{t:\,\mathrm{FPR}(t)\leq\gamma}
      \mathrm{TPR}(t),
\end{equation}
which selects the operating point with the highest empirical TPR while keeping the empirical FPR at or below $\gamma$. All detection metrics are reported as percentages, and higher values indicate better performance.

\paragraph{Generation quality.}
We use perplexity to study the quality--detectability trade-off. For a generated continuation $y=(y_1,\ldots,y_T)$ conditioned on prompt $x$, perplexity under a fixed reference language model is
\begin{equation}
    \mathrm{PPL}(y\mid x)
    =\exp\!\left(
        -\frac{1}{T}\sum_{t=1}^{T}
        \log P_{\mathrm{ref}}(y_t\mid x,y_{<t})
    \right).
\end{equation}
Lower perplexity indicates that the continuation receives higher likelihood under the reference model.

\subsection{Implementation Details}
\label{app:implementation}

\subsubsection{Baseline Methods}
\label{app:imp-baseline}
We compare against KGW, EXP, and PRC, with PRC serving as the closest coding-based baseline to \method. KGW and EXP use the implementations provided by MarkLLM~\citep{pan2024markllm}, while we reuse an open-source implementation of PRC\footnote{https://github.com/patrickrchao/watermarking-llms}. Within each experimental condition, all methods use the same prompt split, base model, tokenizer, maximum generation length, attack configuration, and metric computation. Low-FPR operating points are obtained from the empirical score distributions using the same procedure for every method.

\paragraph{KGW.}
KGW is a token-level green-list watermark that embeds a statistical bias during generation. At each decoding step, the preceding token context is combined with a secret key to seed a pseudorandom partition of the vocabulary. A fraction $\gamma$ of the vocabulary forms the green list, and a constant logit bias $\delta$ is added to its tokens before sampling. We use $\gamma=0.5$, $\delta=2.0$, and \texttt{prefix\_length=1}.

During detection, the continuation is re-tokenized and the keyed green list is reconstructed at every scored position. If $G$ of the $T$ scored tokens belong to their corresponding green lists, KGW computes
\begin{equation}
    z=\frac{G-\gamma T}{\sqrt{T\gamma(1-\gamma)}}.
\end{equation}
The configured fixed decision threshold is $z=4.0$. For AUROC and TPR@FPR, the raw $z$-score is used as the continuous detector score, with larger values indicating stronger watermark evidence.

\paragraph{EXP.}
EXP is a sampling-based watermark that modifies token selection without adding a fixed logit bias. At each generation step, it derives a keyed pseudorandom vector $u\in(0,1)^{|\mathcal{V}|}$ from the preceding tokens and uses this vector together with the model probabilities to perform exponential sampling. Our configuration uses \texttt{prefix\_length=4}.

Detection reconstructs the same pseudorandom vector at every scored position. For an observed token $x_i$, EXP accumulates the statistic
\begin{equation}
    S=\sum_i \log\frac{1}{1-u_i(x_i)}.
\end{equation}
Under the null hypothesis, the tail probability of $S$ is computed using a Gamma distribution whose shape is the number of scored tokens. The configured fixed threshold is $p=10^{-4}$. For the common ROC evaluation, we use $-\log_{10}(p)$ as the continuous detector score, so larger values indicate stronger watermark evidence.

\paragraph{PRC.}
Our PRC baseline is implemented for a pseudorandom-code watermark. The generator first applies a fixed keyed permutation to the vocabulary and represents the permuted token positions in a binary prefix space. A keyed pseudorandom bit sequence is encoded by a small linear error-correcting code, and the resulting encoded bits are used to perturb the model's next-token distribution. We use a fixed vocabulary permutation, allowing the detector to reconstruct the required mapping from the final text alone.

The implementation uses a lightweight Low-Density Parity-Check (LDPC)-style linear code. It constructs a regular parity-check matrix over $\mathbb{F}_2$, then obtains a generator matrix from the null space of the parity-check matrix.

During detection, the continuation is re-tokenized, mapped through the inverse vocabulary permutation, and converted into binary token-position representations. The detector extracts leading bits from token windows and compares them with the expected encoded key bits. It applies a one-sided binomial test against a null match probability of $1/2$.

\begin{table*}[t]
\centering
\small
\setlength{\tabcolsep}{8pt}
\renewcommand{\arraystretch}{1.08}
\begin{tabular}{@{}lccccc@{}}
\toprule
\textbf{Detector variant}
& \textbf{Clean}
& \textbf{Rand. Del. (0.2)}
& \textbf{Syn. Sub. (0.5)}
& \textbf{Burst Del. (0.5)}
& \textbf{Emoji (1.0)} \\
\midrule
Full \method
& \textbf{99.8}
& 90.0
& \textbf{86.6}
& \textbf{98.1}
& \textbf{95.5} \\

Fixed contiguous ($\tau=1$)
& 99.0
& \textbf{92.6}
& 84.3
& \textbf{98.1}
& 86.1 \\

Exact RS consistency ($e=0$)
& 94.1
& 29.7
& 25.8
& 74.3
& 26.2 \\
\bottomrule
\end{tabular}
\caption{
Ablation of the \method detector on C4 with Llama-3.1-8B-Instruct.
We report TPR@1\%FPR, expressed as percentages.
The full detector uses adaptive pilot selection, multiple candidate strides, and Berlekamp--Welch consistency testing with $e=2$.
The fixed-contiguous variant uses only stride $\tau=1$, while the exact-RS variant requires zero symbol mismatches.
The value in parentheses denotes the attack rate.
}
\label{tab:detector-ablation}
\vspace{-2mm}
\end{table*}

\subsubsection{\method Implementation}
\label{app:drw}
For implementation of \method, the generator uses a bivariate Reed--Muller construction over $\mathbb{F}_{17}$ with field size $q=17$, degree bound $d=3$, and local window length $m=12$. We use logit bias $\delta=5.0$ for our experiments.

During detection, the suspect text is tokenized and mapped to its keyed bin symbols. The detector samples local windows at candidate strides and checks whether each window agrees, up to a bounded number of symbol errors, with the evaluations of a univariate polynomial of degree at most $d$. We use a Berlekamp--Welch-style consistency test with \texttt{rs\_errors=2}, allowing up to two mismatches in each length-$m$ window.

The detector uses adaptive stride selection. It first performs 500 consistency trials over candidate strides, then retains at most 8 promising strides for the full evaluation. For each retained stride, the number of successful windows is converted into an approximate binomial-tail score using a stride-specific null success rate. The best score is adjusted for the number of tested strides to form $p_{\mathrm{det}}$.

The stride-specific null rates are calibrated offline. These rates parameterize the binomial null model used by the detector. Accordingly, we refer to $p_{\mathrm{det}}$ as an approximate detection score rather than an exact calibrated $p$-value.

\subsubsection{Computational Resources}
All experiments were conducted on a Linux server running Ubuntu 20.04.6 LTS. The machine is equipped with two AMD EPYC 7H12 64-Core processors, with 255 logical CPUs, and 503~GiB of system memory. GPU-accelerated experiments were run on a single NVIDIA A40 GPU with 46~GB of GPU memory.

\begin{table*}[t]
\centering
\small
\setlength{\tabcolsep}{10pt}
\renewcommand{\arraystretch}{1.12}

\begin{threeparttable}
\begin{tabular}{@{}lcccc@{}}
\toprule
\textbf{Method}
&
\textbf{WM Pass@1}
&
\textbf{AUROC}
&
\textbf{TPR@1\%FPR}
&
\textbf{TPR@5\%FPR}
\\
\midrule

KGW
& 35.7
& 85.8
& 12.8
& 48.6
\\

EXP
& \textbf{36.6}
& \textbf{90.7}
& 37.0
& 63.4
\\

PRC
& 13.3
& 79.7
& 30.1
& 48.3
\\

\method{} (Ours)
& 34.5
& 87.2
& \textbf{37.2}
& \textbf{74.6}
\\

\bottomrule
\end{tabular}
\end{threeparttable}
\caption{
Results on the MBPP benchmark with Llama-3.1-8B-Instruct.
The non-watermarked model obtains a Pass@1 of \textbf{46.7\%}.
All remaining values are percentages, and higher values are better.
}
\label{tab:code-generation}
\vspace{-2mm}
\end{table*}
\begin{table*}[t]
\centering
\normalsize
\renewcommand{\arraystretch}{1.12}

\begin{tabular*}{0.9\textwidth}{
    @{\extracolsep{\fill}}
    lccccc
    @{}
}
\toprule
\textbf{Attack}
& \textbf{10\%}
& \textbf{20\%}
& \textbf{30\%}
& \textbf{40\%}
& \textbf{50\%} \\
\midrule
Rand. Del. & 98.8 & 90.0 & 79.5 & 64.9 & 39.1 \\
Burst Del. & \textbf{99.8} & 99.6 & 98.8 & 99.1 & 98.1 \\
Trunc.     & \textbf{99.8} & \textbf{99.8} & \textbf{99.8}
           & \textbf{99.5} & \textbf{98.9} \\
Emoji      & 98.9 & 98.2 & 97.4 & 97.9 & 96.7 \\
Syn. Sub.  & 99.0 & 98.2 & 95.1 & 90.9 & 86.6 \\
\bottomrule
\end{tabular*}

\caption{Detection performance of \method under varying attack rates. We report TPR@1\%FPR, expressed as percentages, on C4 with Llama-3.1-8B-Instruct. For random deletion, burst deletion, and truncation, the attack rate denotes the fraction of removed content. For emoji attack and synonym substitution, it denotes the emoji insertion probability and maximum word-replacement fraction, respectively. The best result at each attack rate is shown in bold.}
\label{tab:deletion_rate}
\vspace{-3mm}
\end{table*}

\subsection{Attacks}
\label{app:attacks}
We evaluate post-processing deletion and rewriting attacks. In the clean setting, denoted \textsc{Clean}, the detector is applied directly to the original generated continuation.

\paragraph{Random deletion.}
Random deletion (\textsc{Rand. Del.}) splits a continuation into whitespace-separated words and independently deletes each word with probability $r$. The main robustness benchmark uses $r=0.2$.

\paragraph{Burst deletion.}
Burst deletion (\textsc{Burst Del.}) removes contiguous semantic units. It first splits the continuation into paragraphs and falls back to sentence-level units when only one paragraph is available. It then removes a random subset of units according to the attack ratio while retaining at least one unit. The main robustness benchmark uses $r=0.5$.

\paragraph{Partial truncation.}
Partial truncation (\textsc{Trunc.}) retains one randomly selected contiguous span and removes the remaining text. With attack ratio $r$, approximately a $1-r$ fraction of the original words is retained. The main robustness benchmark uses $r=0.5$.

\paragraph{Emoji attack.}
The emoji attack (\textsc{Emoji.}) inserts a randomly sampled emoji after each word with independent probability $r$. The main robustness benchmark uses $r=1.0$, which inserts an emoji after every word and substantially changes tokenization while preserving the visible lexical content. We group it with deletion-based attacks in this work because they both modify the length of token sequence and shift the positions of all subsequent tokens, inducing synchronization loss.

\paragraph{Synonym substitution.}
Synonym substitution (\textsc{Syn. Sub.}) uses WordNet~\citep{miller1995wordnet} to identify replaceable words, randomly selects up to an $r$ fraction of the words, and replaces each selected word with a randomly sampled synonym lemma when available. The main robustness benchmark uses $r=0.5$.

\paragraph{Paraphrasing.}
Paraphrasing (\textsc{Para.}) uses the document-level editor provided by MarkLLM~\citep{pan2024markllm}. The reported results use its DIPPER-based paraphraser~\citep{krishna2023paraphrasing}.

\section{Further Analysis}
\label{app:further_analysis}

\subsection{Ablation Study}
\label{app:ablation}

We study the contributions of two components in the \method detector: adaptive search over candidate strides and tolerance to symbol mismatches in the local Reed--Solomon consistency test. The full detector uses a pilot stage to identify promising candidate strides, searches multiple strides, and applies Berlekamp--Welch decoding with a mismatch budget of $e=2$. The \emph{fixed contiguous} variant restricts the detector to contiguous windows with $\tau=1$, while retaining the same mismatch budget. The \emph{exact RS} variant keeps the candidate-stride search but requires exact degree-$d$ consistency by setting $e=0$.

Table~\ref{tab:detector-ablation} shows that the full detector obtains the best or tied-best performance in four of the five settings. Compared with restricting the detector to contiguous windows, the full method improves TPR by 2.3 points under synonym substitution and by 9.4 points under emoji insertion, while maintaining comparable performance on clean text and burst deletion. The particularly large gain under the emoji attack indicates that searching multiple candidate strides is useful when post-processing substantially changes the tokenization and disrupts the local spacing of the induced symbol sequence.

For burst deletion, both variants achieve 98.1\% TPR. This is consistent with the structure of the attack: although a contiguous region is removed, the remaining text can still contain long, unchanged spans in which stride-$1$ windows preserve local algebraic structure. Interestingly, the fixed-contiguous variant performs 2.6 points better under random deletion. This result indicates that adaptive stride search is not uniformly advantageous across all deletion patterns. Its empirical benefit depends on how the attack alters the surviving symbol sequence.

The mismatch budget is substantially more important. Requiring exact RS consistency reduces clean text TPR from 99.8\% to 94.1\%, even without post-processing. This decrease arises because logit biasing encourages, but does not force, each generated token to belong to its target vocabulary bin. The realized symbol sequence therefore contains occasional mismatches with the target RM structure even before an attack.

The effect becomes much larger after post-processing. Compared with the full detector, exact consistency decreases TPR by 60.3 points under random deletion, 60.8 points under synonym substitution, 23.8 points under burst deletion, and 69.3 points under emoji insertion. These results demonstrate that bounded-error consistency testing is essential for absorbing both intrinsic generation noise and symbol mismatches introduced by text editing. Overall, the ablation supports the use of Berlekamp--Welch decoding as a central component of the detector, while showing that adaptive multi-stride search provides additional, attack-dependent robustness.

\begin{table*}[t]
\centering
\small
\setlength{\tabcolsep}{3.8pt}
\renewcommand{\arraystretch}{1.08}
\begin{tabular}{@{}lccc@{}}
\toprule
\textbf{Method}
& \textbf{Gen. speed (tokens/s)} $\uparrow$
& \textbf{Detect. time (s/text)} $\downarrow$
& \textbf{Avg. gen. time (s/sample)} $\downarrow$ \\
\midrule
KGW
& \textbf{59.8}
& 0.95
& \textbf{8.2} \\
EXP
& 15.0
& 0.26
& 33.3 \\
PRC
& 11.1
& \textbf{0.01}
& 42.4 \\
\method\ (Ours)
& 16.4
& 0.57
& 30.5 \\
\bottomrule
\end{tabular}
\caption{
Wall-clock efficiency of the evaluated watermarking methods under a
maximum generation length of 500 tokens. Generation speed is measured
in generated tokens per second, detection time is the average latency
per input text, and generation time is the average latency per sample.
All methods are measured under the same hardware and runtime
configuration. Higher generation speed and lower latency are better.
}
\label{tab:efficiency}
\vspace{-2mm}
\end{table*}

\subsection{Code Generation Task}
\label{app:code-generation}

We evaluate whether the watermarking methods transfer beyond natural-language generation to code generation using the Mostly Basic Python Problems (MBPP) benchmark~\citep{austin2021program} and Llama-3.1-8B-Instruct.

\paragraph{Dataset and evaluation setup.} MBPP contains 974 crowdsourced Python programming tasks designed to be solvable by entry-level programmers. Each task provides a natural-language problem description, a reference solution, and three executable test cases for checking functional correctness. The tasks cover common programming concepts, including numerical operations, string and list manipulation, control flow, and basic use of the Python standard library. We report Pass@1 as the percentage of generated programs that pass all provided test cases, together with AUROC, TPR@1\%FPR, and TPR@5\%FPR for watermark detection.

\paragraph{Code-generation utility.} As shown in Table~\ref{tab:code-generation}, all evaluated watermarking methods reduce Pass@1 relative to the unwatermarked model. EXP obtains the highest watermarked Pass@1 of 36.6\%, followed by KGW at 35.7\% and \method at 34.5\%. \method therefore incurs a 12.2-point reduction relative to the unwatermarked model, while remaining within 2.1 points of the strongest watermarked result. PRC exhibits a substantially larger utility decrease, achieving only 13.3\% Pass@1. These results indicate that imposing a watermark can affect exact functional correctness in code generation, although \method preserves utility at a level comparable to KGW and EXP.

\paragraph{Watermark detectability.} EXP achieves the highest overall AUROC of 90.7\%, while \method obtains 87.2\%, outperforming KGW and PRC by 1.4 and 7.5 points, respectively. The advantage of \method becomes more evident at low false-positive rates. At 1\% FPR, \method achieves the highest TPR of 37.2\%, slightly exceeding EXP at 37.0\% and outperforming KGW and PRC by 24.4 and 7.1 points. At 5\% FPR, \method reaches 74.6\% TPR, improving over EXP by 11.2 points and over KGW and PRC by 26.0 and 26.3 points, respectively.

Overall, \method provides the strongest low-FPR watermark detection
among the evaluated methods while maintaining code-generation utility
comparable to KGW and EXP. At the same time, the decrease from the unwatermarked Pass@1 highlights a utility--detectability trade-off and motivates future work on watermark injection strategies that better preserve exact program correctness.

\subsection{Deletion Rate}
We further examine how \method responds to increasing post-processing strength. Table~\ref{tab:deletion_rate} reports TPR@1\%FPR as the attack
rate increases from 10\% to 50\%. Across structured deletion attacks, \method remains highly stable. Under burst deletion, TPR stays above 98\% at every evaluated rate and reaches 98.1\% even when half of the semantic units are removed. Truncation has an even smaller effect, with TPR decreasing only from 99.8\% to 98.9\%. TPR remains above 96\% across all evaluated emoji insertion rates. These results support the central design intuition of \method. Although a large contiguous portion may be removed, the surviving text preserves sufficiently many ordered local subsequences for the detector to identify Reed--Solomon consistency.

Random deletion produces a more gradual decline as independently removed words disperse synchronization changes throughout the continuation. Nevertheless, \method retains 98.8\% TPR at a 10\% deletion rate and 90.0\% TPR at a 20\% deletion rate. The comparison with burst deletion and truncation demonstrates the benefit of searching for surviving local algebraic structure rather than requiring recovery of a globally aligned codeword.

The detector also remains stable when the attack modifies rather than removes tokens. Under synonym substitution, TPR is 95.1\% at a 30\% replacement rate and remains 86.6\% when up to half of the words are replaced. Together, these results demonstrate that \method preserves strong detectability across a wide range of perturbation types and attack strengths.

\subsection{Watermarking Efficiency}
\label{app:efficiency}

\paragraph{Measurement setup.}
We measure all methods under the same hardware and runtime configuration. Experiments are conducted using one NVIDIA A40 GPU with 46\,GB of available memory and two AMD EPYC 7H12 processors, providing 128 physical cores and 255 logical CPUs. We set \texttt{OMP\_NUM\_THREADS=32}. The generation experiments use a maximum output length of 500 tokens. Generation throughput and per-sample generation latency are measured over the same benchmark examples.

\paragraph{Efficiency analysis.}
Table~\ref{tab:efficiency} compares the generation and detection efficiency of \method with the three baseline methods. \method achieves a generation throughput of 16.4 tokens per second, the second highest among the evaluated methods. Its throughput is approximately 9.3\% higher than EXP and 47.7\% higher than PRC. Accordingly, its average generation latency of 30.5 seconds per sample is 8.4\% lower than EXP and 28.1\% lower than PRC. KGW remains substantially faster during generation, achieving 59.8 tokens per second and an average latency of 8.2 seconds per sample.

For detection, \method requires 0.57 seconds per text. This is 40\% lower than the 0.95-second latency of KGW, although it is higher than the latency of EXP and PRC. The additional detection cost is consistent with the design of \method: the detector first searches over candidate strides and then applies repeated Reed--Solomon consistency tests to sampled local windows. In contrast, the baseline detectors use less expensive aggregate statistics or direct code-matching procedures.

Importantly, the detection latency of \method remains below one second and corresponds to less than 2\% of its average generation latency. Thus, although local algebraic testing introduces additional detection work relative to EXP and PRC, detection does not dominate the overall runtime. Taken together, the results show that \method provides a balanced efficiency profile: it generates text faster than the sampling- and coding-based baselines EXP and PRC, while retaining sub-second detection latency.

\section{Case Study}
\label{app:case}
We present four representative case studies to qualitatively examine the behavior of \method and the evaluated post-processing attacks. The first three cases compare continuations generated by KGW, EXP, PRC, and \method from the same C4 prompts. All methods use Llama-3.1-8B-Instruct and the same generation configuration. The examples cover three different discourse settings: a local community complaint (\Sref{case:local-complaint}), an expository discussion (\Sref{case:public-marketing}), and a reflective essay (\Sref{case:poetry}).

The fourth case fixes one clean continuation generated by \method and applies six benchmark attacks to it (\Sref{case:attacks}). This controlled comparison illustrates how random deletion, burst deletion, partial truncation, emoji attack, synonym substitution, and paraphrasing affect the visible text and the underlying token-derived watermark sequence in different ways.

\definecolor{casegray}{HTML}{555555}
\definecolor{casegraybg}{HTML}{F5F5F5}
\definecolor{caseblue}{HTML}{527A9E}
\definecolor{casebluebg}{HTML}{F2F6FA}
\definecolor{casegreen}{HTML}{3E8C75}
\definecolor{casegreenbg}{HTML}{F2F8F5}
\definecolor{caseorange}{HTML}{A56B32}
\definecolor{caseorangebg}{HTML}{FBF6F0}

\newtcolorbox{casepromptbox}{
  enhanced,
  breakable,
  colback=casegraybg,
  colframe=casegray!65,
  colbacktitle=casegray!85,
  coltitle=white,
  fonttitle=\bfseries,
  title={Prompt},
  boxrule=0.4pt,
  arc=1pt,
  left=5pt,
  right=5pt,
  top=4pt,
  bottom=4pt,
  before upper=\small\itshape\raggedright,
  before skip=4pt,
  after skip=4pt
}

\newtcolorbox{caseoutputbox}[1]{
  enhanced,
  breakable,
  colback=casebluebg,
  colframe=caseblue,
  colbacktitle=caseblue,
  coltitle=white,
  fonttitle=\bfseries,
  title={#1},
  boxrule=0.4pt,
  arc=1pt,
  left=5pt,
  right=5pt,
  top=4pt,
  bottom=4pt,
  before upper=\small\raggedright,
  before skip=4pt,
  after skip=4pt
}

\newtcolorbox{caseoursbox}{
  enhanced,
  breakable,
  colback=casegreenbg,
  colframe=casegreen,
  colbacktitle=casegreen,
  coltitle=white,
  fonttitle=\bfseries,
  title={\method\ (Ours)},
  boxrule=0.5pt,
  arc=1pt,
  left=5pt,
  right=5pt,
  top=4pt,
  bottom=4pt,
  before upper=\small\raggedright,
  before skip=4pt,
  after skip=4pt
}

\newtcolorbox{caseanalysisbox}{
  enhanced,
  breakable,
  colback=caseorangebg,
  colframe=caseorange,
  colbacktitle=caseorange,
  coltitle=white,
  fonttitle=\bfseries,
  title={Case Analysis},
  boxrule=0.5pt,
  arc=1pt,
  left=5pt,
  right=5pt,
  top=4pt,
  bottom=4pt,
  before upper=\small\raggedright,
  before skip=4pt,
  after skip=4pt
}

\newtcolorbox{casecleanbox}{
  enhanced,
  breakable,
  colback=casegreenbg,
  colframe=casegreen,
  colbacktitle=casegreen,
  coltitle=white,
  fonttitle=\bfseries,
  title={Clean \method\ Output},
  boxrule=0.5pt,
  arc=1pt,
  left=5pt,
  right=5pt,
  top=4pt,
  bottom=4pt,
  before upper=\small\raggedright,
  before skip=4pt,
  after skip=5pt
}

\newtcolorbox{attacktracebox}[1]{
  enhanced,
  breakable,
  colback=casebluebg,
  colbacklower=caseorangebg,
  colframe=caseblue,
  colbacktitle=caseblue,
  coltitle=white,
  fonttitle=\bfseries,
  title={#1},
  boxrule=0.4pt,
  arc=1pt,
  left=5pt,
  right=5pt,
  top=4pt,
  bottom=4pt,
  segmentation style={
    solid,
    caseblue!45,
    line width=0.4pt
  },
  before upper=\small\raggedright,
  before lower=\small\raggedright,
  before skip=5pt,
  after skip=5pt
}

\subsection{Case 1: Local Community Complaint}
\label{case:local-complaint}

\begin{casepromptbox}
RANT: To the irresponsible young adults who park their vehicles at
Winter Haven Christian Church on Sixth Street and then leave trash and
empty beer bottles for others
\end{casepromptbox}

\begin{caseoutputbox}{KGW}
You know who you irresponsible people who leave their cars on our
neighborhood's most holy ground, you need to take them there. to clean
up. What is wrong with you people. Do you have no sense of decency or
respect for other people's property. I know you have no respect for
yourselves either, if you can't even take care of your own trash. Get
it together and take care of yourself. \ldots
\end{caseoutputbox}

\begin{caseoutputbox}{EXP}
to clean up. The Rant: To the irresponsible young adults who park their
vehicles at Winter Haven Christian Church on Sixth Street and then
leave trash and empty beer bottles for others to clean up. The Rant:
To the irresponsible young adults who park their vehicles at Winter
Haven Christian Church on Sixth Street and then leave trash and empty
beer bottles for others to clean up. \ldots
\end{caseoutputbox}

\begin{caseoutputbox}{PRC}
to clean up! ANOTHER RANT: To all of the people who think it's okay to
litter in our beautiful city! I can understand why you might think it's
convenient to throw your trash out the window of your car or leave it
on the sidewalk, but when you litter, you are contributing to the
degradation of our city and making it less enjoyable for everyone.
\ldots
\end{caseoutputbox}

\begin{caseoursbox}
to clean up\ldots{} I live a short drive from Winter Haven Christian
where I often notice groups of young adults leave their cars and attend
events at that church. While I support the church's various gatherings,
I have noticed an alarming and irresponsible trend. I, along with other
neighbors, see evidence of trash left behind---such as empty beer
bottles, cigarette butts, cigarette packs and sometimes clothing and
luggage tossed around and trashed. \ldots
\end{caseoursbox}

\begin{caseanalysisbox}
The continuation produced by \method preserves the first-person,
community-complaint perspective introduced by the prompt and develops
it using location-specific details. KGW remains on topic but contains
grammatical discontinuities and shifts toward a substantially more
confrontational tone. EXP degenerates into repeated reconstruction of
the prompt. PRC produces a coherent anti-littering discussion, but
moves from the specific church and neighborhood setting toward a more
generic civic message. This example illustrates that \method can retain
the local voice, topic, and communicative purpose of the source text
while embedding the watermark.
\end{caseanalysisbox}
\subsection{Case 2: Public-Sector Marketing}
\label{case:public-marketing}

\begin{casepromptbox}
Marketers are employed in the public sector as well as the private
sector. Most people think of marketing as a strictly private sector
activity, but
\end{casepromptbox}

\begin{caseoutputbox}{KGW}
a. They work in advertising agencies and marketing firms. They work
for both products and services are marketed by companies that operate
for profit. But marketing is also a key function in government
agencies and non-profit organizations. These agencies and organizations
also have to find ways to persuade and inform the public about their
products and services. \ldots
\end{caseoutputbox}

\begin{caseoutputbox}{EXP}
many governments offer marketing services to promote economic
development. The U.S. Department of Commerce, for example, has a
Bureau of Industry and Security that markets U.S. goods and services
abroad. State and local governments also have tourism offices that
market vacation destinations and local businesses. \ldots
\end{caseoutputbox}

\begin{caseoutputbox}{PRC}
that is not true. Many public sector organizations, such as government
agencies, public hospitals, and non-profit organizations, use marketing
to achieve their goals. The main difference between marketing in the
private and public sectors is that public sector organizations do not
aim to maximize profits, but rather to achieve social or public goals.
\ldots
\end{caseoutputbox}

\begin{caseoursbox}
the truth is that marketing exists in the public and non-profit sectors
as well. Federal, state, and local governments, as well as non-profit
agencies employ market researchers and analysts to fulfill various
requirements, such as those for policy and program support, budget
formulation, service innovation, and community engagement. Those in
this field make up a crucial component of any sector, including
not-for-profit organizations, which promote important social causes.
\ldots
\end{caseoursbox}

\begin{caseanalysisbox}
All four methods remain broadly relevant to the prompt, but differ in
how naturally they continue its unfinished contrast. KGW begins with an
incomplete construction and contains grammatical problems. EXP and PRC
produce coherent explanations, although PRC adopts a relatively generic
textbook-style comparison between public and private organizations.
\method directly completes the contrast introduced by the prompt,
maintains an article-like explanatory register, and provides concrete
examples involving policy support, budget formulation, service
innovation, and community engagement.
\end{caseanalysisbox}
\subsection{Case 3: Poetry, War, and Witness}
\label{case:poetry}

\begin{casepromptbox}
What can poetry do when faced with the daily realities of war and
hatred, violence and terror, the crimes that humans enact upon each
other and on
\end{casepromptbox}

\begin{caseoutputbox}{KGW}
What is poetry do in these circumstances? What can poetry offer the
world? Poetry can speak out against humanity. Poetry can give voice to
the voiceless, and it can speak for those who are oppressed. Poetry
can stand witness to the crimes against humanity. Poets can speak for
the silenced. Poetry can help to humanize the victims and to give
dignity to the dead. \ldots
\end{caseoutputbox}

\begin{caseoutputbox}{EXP}
the earth? Where can we turn in our faith, our hope, and our will to
live in the face of such overwhelming darkness? These are the questions
that poet and activist Judy Grahn poses in her new collection,
\emph{Love Belongs to Those Who Fight and Die}, a powerful exploration
of the transformative power of art and love to resist hatred and
violence. \ldots
\end{caseoutputbox}

\begin{caseoutputbox}{PRC}
the earth? Can it inspire resistance, comfort, and healing? And what
about its relationship to politics, protest, and activism? These are
some of the questions that poet and translator Ilya Kaminsky explores
in this anthology, which brings together poems from twenty different
languages and poets from around the world. \ldots
\end{caseoutputbox}

\begin{caseoursbox}
the earth? It can try to bear witness; to make public the private pains
and losses; to seek justice, or perhaps even imagine it; to preserve
beauty in the face of destruction; to analyze the circumstances behind
human hatred; or to imagine a different world, one in which love
conquers all. It can express the complicated nature of war---both the
horrors it creates and the power and vulnerability of its victims.
\ldots
\end{caseoursbox}

\begin{caseanalysisbox}
The prompt poses an abstract question and invites a reflective
continuation. \method preserves this discourse mode by directly
enumerating the roles poetry may play, including witnessing suffering,
preserving beauty, seeking justice, and imagining alternatives. KGW
remains on topic but contains an immediate grammatical error and
repetitive formulations. EXP and PRC are fluent, but both shift the
abstract discussion toward a description of a particular poet, book,
or anthology. The \method output instead continues the original
essay-level argument without replacing it with a more specific and
unprompted framing.
\end{caseanalysisbox}
\subsection{Case 4: Effects of Different Post-Processing Attacks}
\label{case:attacks}

All attacked versions below are derived from the same clean \method-generated text.

\begin{casepromptbox}
Air strikes conducted near Erbil and the Mosul dam amid reports of mass
killings by Sunni fighters in Syria. The United States has conducted
air strikes
\end{casepromptbox}

\begin{casecleanbox}
against Sunni militants near the north-eastern Iraqi city of Erbil and
along the Tigris River at the Mosul dam. The US defence department said
the strikes were carried out at President Barack Obama's direction on
Thursday to protect civilians under attack from the militant group,
which calls itself Islamic State. The targets were struck around Erbil,
the location of Kurdish leaders including Iraq's president, Massud
Barzani, and the US's largest diplomatic mission in Iraq. There was no
independent confirmation of the strikes, which was based on reports
from a US official at the Pentagon. \ldots{} Separately, reports have
emerged of mass killings in Syria.
\end{casecleanbox}

\begin{attacktracebox}{
Random Deletion
\hfill
\textnormal{\textsc{Rand. Del.}, $r=0.2$}
}
\textbf{Attacked Text.}
Sunni militants the north-eastern Iraqi city of Erbil and along the
Tigris River at the Mosul. The defence said the strikes were carried
President Barack Obama's direction on Thursday protect civilians under
from the militant group, which calls itself Islamic State. The struck
Erbil, the location of leaders Iraq's president, Massud Barzani, the
US's diplomatic mission in Iraq. There was no confirmation of the
strikes, which based on the reports from a official at the Pentagon.
\ldots

\tcblower

\textbf{Qualitative Effect.}
Words are removed independently throughout the continuation. Although
many content words remain, missing articles, prepositions, and
connectors damage grammaticality. Each deletion also shifts the
positions of subsequent token-derived symbols, producing distributed
synchronization changes across the text.
\end{attacktracebox}

\begin{attacktracebox}{
Burst Deletion
\hfill
\textnormal{\textsc{Burst Del.}, $r=0.5$}
}
\textbf{Attacked Text.}
There was no independent confirmation of the strikes, which was based
on the reports from a US official at the Pentagon, as the militants,
who are also known as Isis, maintain control of most media outlets in
the region and telephone connectivity was down. Separately, reports
have emerged of mass killings in Syria. Turkey confirmed earlier
Thursday that it would be taking part in a new co-ordinated effort by
Sunni Gulf states. \ldots

\tcblower

\textbf{Qualitative Effect.}
Burst deletion removes contiguous sentence- or paragraph-level units.
The beginning of the report is absent, but the retained passage remains
largely fluent because the internal order of each surviving span is
unchanged. The detector therefore loses a large region while still
observing long locally ordered subsequences.
\end{attacktracebox}

\begin{attacktracebox}{
Partial Truncation
\hfill
\textnormal{\textsc{Trunc.}, $r=0.5$}
}
\textbf{Attacked Text.}
from the militant group, which calls itself Islamic State. The targets
were struck around Erbil, the location of Kurdish leaders including
Iraq's president, Massud Barzani, and the US's largest diplomatic
mission in Iraq. There was no independent confirmation of the strikes,
which was based on the reports from a US official at the Pentagon.
\ldots{} Separately, reports have emerged of mass killings in Syria.

\tcblower

\textbf{Qualitative Effect.}
Partial truncation retains one contiguous subsection of the original
continuation. The passage begins in the middle of a sentence, but its
internal order and local syntax remain mostly unchanged. The primary
challenge is the reduced amount of available watermark evidence rather
than distributed corruption.
\end{attacktracebox}

\begin{attacktracebox}{
Emoji Attack
\hfill
\textnormal{\textsc{Emoji.}, $r=1.0$}
}
\textbf{Attacked Text.}
against \texttt{[emoji]} Sunni \texttt{[emoji]} militants
\texttt{[emoji]} near \texttt{[emoji]} the \texttt{[emoji]}
north-eastern \texttt{[emoji]} Iraqi \texttt{[emoji]} city
\texttt{[emoji]} of \texttt{[emoji]} Erbil \texttt{[emoji]} and
\texttt{[emoji]} along \texttt{[emoji]} the \texttt{[emoji]} Tigris
\texttt{[emoji]} River \texttt{[emoji]} at \texttt{[emoji]} the
\texttt{[emoji]} Mosul \texttt{[emoji]} dam. \texttt{[emoji]} The
\texttt{[emoji]} US \texttt{[emoji]} defence \texttt{[emoji]}
department \texttt{[emoji]} said \texttt{[emoji]} the
\texttt{[emoji]} strikes \texttt{[emoji]} were carried out. \ldots

\tcblower

\textbf{Qualitative Effect.}
The original words and their relative order remain almost unchanged,
so the semantic content remains recognizable to a reader. However,
inserting an emoji after every word substantially changes tokenization
and the spacing between watermark-bearing symbols. This attack
therefore creates severe token-level disruption without semantic
rewriting.
\end{attacktracebox}

\begin{attacktracebox}{
Synonym Substitution
\hfill
\textnormal{\textsc{Syn. Sub.}, $r=0.5$}
}
\textbf{Attacked Text.}
against Sunni activist near the north-eastern Iraki urban center of
Erbil and along the Tigris River River atomic number 85 the Mosul dam.
The US defense section said the strikes be transmit come out atomic
number 85 President Barack Obama's direction on Th to protect civilians
below attack from the warlike group, which calls itself Islamic State.
The object live move around Erbil, the locating of Kurdish leaders let
in Iraq's president, Massud Barzani. \ldots

\tcblower

\textbf{Qualitative Effect.}
The attack replaces many lexical items with WordNet alternatives.
Because the replacements are not fully context-aware, some are
inappropriate; for example, the preposition ``at'' is replaced by
``atomic number 85.'' The operation changes the vocabulary-bin
assignments of many observed tokens while also introducing semantic
and grammatical artifacts.
\end{attacktracebox}

\begin{attacktracebox}{
Document-Level Paraphrasing
\hfill
\textnormal{\textsc{Para.}}
}

\smallskip
\textbf{Attacked Text.}
North-East of Iraq, near the Mosul dam, against Sunni fighters in Iraq,
the US defense department has said. President Obama ordered the
strikes, it said, to protect civilians under attack from the militant
group which calls itself the Islamic State. The target was around
Erbil, where the Kur theorist leader Barham Salih is holding a meeting,
as well as other sites. Reports are coming in of mass killings by Sunni
fighters in Syria. The UN human rights spokesman, Rupert Colville, said
that Sunni fighters had killed hundreds in Russia and Iraq. Iran has
offered hundreds of troops to help the government forces, but the
number has not been confirmed, because of a dispute with Iraq about
the strategic importance and tactical control of this help. Egypt has
also agreed to help, sending about 170 elite commandos and fighter jets
against the main targets. \ldots

\tcblower

\textbf{Qualitative Effect.}
DIPPER produces a fluent document-level paraphrase that preserves the
broad topic of US air strikes near Erbil and reports of violence in
Syria. However, it substantially rewrites the lexical realization and
sentence structure of the original continuation. It also introduces
factual drift and entity errors in this example, including altered
named entities, incorrect locations, and additional details not
supported by the clean output. Compared with word-level substitution,
the resulting text is more readable, but substantially less of the
original token-derived watermark evidence is expected to remain.
\end{attacktracebox}

\begin{caseanalysisbox}
The example shows that the fraction of modified words alone does not fully characterize attack difficulty. Burst deletion and partial truncation remove substantial amounts of content but preserve long ordered spans. Random deletion removes a smaller fraction of the text but distributes synchronization shifts throughout the sequence. Emoji insertion preserves nearly all visible words while heavily perturbing tokenization, whereas synonym substitution directly changes the lexical items from which the watermark symbols are obtained.\\

Document-level paraphrasing presents a qualitatively different challenge. Unlike the other attacks, DIPPER generates a new realization of the passage rather than locally editing the original token sequence. The paraphrase remains fluent and preserves the broad news topic, but changes much of the wording and sentence structure. It also introduces semantic drift and factual artifacts in this example. Consequently, there may be few directly surviving local symbol sequences even when the paraphrased text remains recognizable as semantically related to the clean output.\\

These differences motivate the two principal components of the \method detector. Multi-stride search helps locate surviving local structure whose spacing has changed, while bounded-error Reed--Solomon consistency testing tolerates a limited number of symbol mismatches. These mechanisms are well suited to attacks that preserve some local lexical structure, including structured deletion, tokenization changes, and moderate substitution. In contrast, document-level paraphrasing may replace most of the original lexical realization and therefore represents a more fundamental limitation of local algebraic detection.
\end{caseanalysisbox}

\end{document}